\documentclass[5p,times,twocolumn]{elsarticle}

\usepackage{amsmath,amssymb,amsthm,mathtools}
\usepackage{booktabs,tabularx,array}
\usepackage{microtype}
\usepackage[T1]{fontenc}
\usepackage{hyperref}
\hypersetup{
  colorlinks=true,
  linkcolor=blue,
  citecolor=blue,
  urlcolor=blue,
  pdftitle={Quota and population monotonicity across house sizes are incompatible for apportionment to four states},
  pdfauthor={Lav R. Varshney}
}

\newtheorem{theorem}{Theorem}
\newtheorem{lemma}{Lemma}
\newtheorem{corollary}{Corollary}
\theoremstyle{definition}
\newtheorem{definition}{Definition}

\newcommand{\1}{\mathbf{1}}
\newcommand{\Znn}{\mathbb{Z}_{\ge 0}}
\newcommand{\Zpos}{\mathbb{Z}_{>0}}
\newcommand{\alloc}[1]{\texttt{#1}}

\journal{}

\begin{document}

\begin{frontmatter}

\title{Quota and population monotonicity across house sizes are incompatible for apportionment to four states}

\author[sbu]{Lav R. Varshney\corref{cor1}}
\ead{lav.varshney@stonybrook.edu}
\cortext[cor1]{Corresponding author.}
\affiliation[sbu]{organization={AI Innovation Institute and Department of Electrical and Computer Engineering, Stony Brook University},
                  addressline={Stony Brook},
                  city={New York},
                  postcode={11794},
                  country={USA}}

\begin{abstract}
Apportionment converts fractional entitlements into integer seat allocations. Quota requires each state to receive the floor or ceiling of its standard quota, whereas population monotonicity prohibits a state whose population weakly increases from losing a seat to a state whose population weakly decreases. G\"olz, Peters, and Procaccia recently removed the order-preservation assumptions used in classical incompatibility results by giving a five-state construction; together with the three-state Webster possibility result, this left four states as the unresolved boundary. We prove that no deterministic four-state apportionment solution satisfies both axioms under their definition, which permits the compared house sizes to differ. The proof is a finite logical gadget. Conditional on one quota choice at a central profile, twelve auxiliary profiles encode three bits and force a frustrated cycle. The argument uses only transfers between states whose populations are unchanged and assumes neither anonymity, neutrality, order preservation, homogeneity, nor other regularity conditions. It also applies to relative population monotonicity. Geometrically, the result is a global compatibility obstruction for quota-constrained lattice rounding, rather than an average-distortion bound.
\end{abstract}


\end{frontmatter}

\section{Introduction}

Apportionment is the institutional problem of converting proportional claims into indivisible representation. Its canonical political application is the allocation of seats in the U.S. House of Representatives among the states after each decennial census, but the same mathematics appears in party-list representation and in the allocation of committee positions or other indivisible public resources \cite{balinski2001fair,pukelsheim2017proportional,robinson2010mathematical}. The problem is not merely numerical. A single seat can alter a state's political influence, and the choice of rule has repeatedly been entangled with partisan incentives, litigation, and disputes over census counts and house size \cite{golz2025lottery}.

Hamilton's largest-remainder rule honors quota: if state $i$ has fractional entitlement $q_i$, it receives either $\lfloor q_i\rfloor$ or $\lceil q_i\rceil$ seats. Yet Hamilton's rule exhibits the Alabama paradox, in which increasing the house can cost a state a seat, and the population paradox, in which representation can move from a faster-growing state to a slower-growing state. Congress adopted the Huntington--Hill divisor method in 1941 in part because it avoids these monotonicity failures, but divisor methods may violate quota. The classical theory therefore exposes a basic political design tension between transparent local proportionality and consistency under changing populations and capacities \cite{balinski1975quota,balinski1979quotatone,balinski1980webster,balinski2001fair}.

The precise scope of the impossibility matters. Classical four-state arguments going back to Balinski--Young use regularity assumptions, usually including the order-preserving requirement that a more populous state not receive fewer seats than a less populous state. Such assumptions are natural for a single deterministic rule but are inappropriate when a randomized method is viewed as a distribution over globally defined deterministic solutions: an individual realization may intentionally favor a smaller state so that the random method is proportionally fair in expectation. G\"olz, Peters, and Procaccia \cite{golz2025lottery} therefore asked whether quota and population monotonicity can coexist without order preservation or neutrality. They settled the general incompatibility by a construction with five states, while the Webster method is population monotone and satisfies quota for three states \cite{balinski1980webster}. Thus four states remained the sharp unresolved boundary.

This note closes that boundary for the definition of population monotonicity used in \cite{golz2025lottery}, under which the two compared house sizes may differ. The proof has three useful features.
\begin{enumerate}
\item It dispenses with anonymity, neutrality, order preservation, homogeneity, divisor structure, and continuity or other regularity assumptions.
\item It replaces the extra combinatorial slack supplied by a fifth state with synchronization across profiles: quota choices are encoded as bits, copied across house sizes, and linked by six implication gadgets.
\item Every implication follows from one elementary prohibition---a seat cannot be transferred between two states whose populations are unchanged. Consequently, the same construction also violates relative population monotonicity. The final contradiction is the frustrated Boolean cycle
\[
A=B,\qquad B=C,\qquad C=1-A.
\]
\end{enumerate}
The proof uses house sizes $3$, $5$, $8$, and $10$; it therefore does not settle the distinct fixed-house version in which population monotonicity is imposed only when the two house sizes are equal.

\section{Model and a transfer lemma}

Let $N=\{1,2,3,4\}$. An input is a pair $x=(p,h)$, where $p=(p_1,p_2,p_3,p_4)\in\Zpos^4$ is a population vector and $h\in\Zpos$ is a house size. An allocation is $a\in\Znn^4$ with $\sum_i a_i=h$. A deterministic \emph{apportionment solution} is a function
\[
f:\Zpos^4\times\Zpos\longrightarrow\Znn^4,
\qquad \sum_{i=1}^4 f_i(p,h)=h.
\]
Writing $P=\sum_i p_i$, the standard quota of state $i$ is
\[
q_i(p,h)=\frac{hp_i}{P}.
\]
The solution satisfies \emph{quota} if
\begin{equation}
 f_i(p,h)\in\{\lfloor q_i(p,h)\rfloor,\lceil q_i(p,h)\rceil\}
 \quad\text{for every }p,h,i.
 \label{eq:quota}
\end{equation}

Following Robinson and Ullman \cite{robinson2010mathematical} and G\"olz et al.\ \cite{golz2025lottery}, two inputs $(p,h)$ and $(p',h')$ and two distinct states $i,j$ exhibit a \emph{population paradox} under $f$ if
\begin{equation}
\begin{aligned}
 p_i'&\ge p_i, & f_i(p',h')&<f_i(p,h),\\
 p_j'&\le p_j, & f_j(p',h')&>f_j(p,h).
\end{aligned}
\label{eq:population-paradox}
\end{equation}

\begin{definition}[Population monotonicity across house sizes]
A solution is population monotone if no comparison satisfying \eqref{eq:population-paradox} exists for any $p,p',h,h',i,j$. In particular, $h$ and $h'$ need not coincide.
\end{definition}

Setting $p'=p$ shows that this axiom also implies house monotonicity. A stronger, relative version declares a paradox whenever state $i$ loses a seat to state $j$ despite
\[
\frac{p_i'}{p_i}\ge \frac{p_j'}{p_j};
\]
that is, $i$ has grown at least as fast as $j$. A solution with no such comparison is \emph{relatively population monotone} \cite{robinson2010mathematical,golz2025lottery}.

For compactness, write $a_1a_2a_3a_4$ for the allocation $(a_1,a_2,a_3,a_4)$. We write $i\to j$ for a comparison in which state $i$ loses one or more seats, state $j$ gains one or more seats, and both populations are unchanged.

\begin{lemma}[Unchanged-state transfer]
Suppose $p_i'=p_i$ and $p_j'=p_j$. Population monotonicity forbids a comparison in which state $i$ loses a seat and state $j$ gains a seat.
\end{lemma}

\begin{proof}
The loss by $i$ satisfies the first line of \eqref{eq:population-paradox}, since equality implies $p_i'\ge p_i$. The gain by $j$ satisfies the second line, since equality implies $p_j'\le p_j$. Hence the comparison is a population paradox.
\end{proof}

\section{A four-state impossibility gadget}

\begin{theorem}
Under Definition~1, which permits the compared house sizes $h$ and $h'$ to differ, no deterministic four-state apportionment solution satisfies both quota and population monotonicity.
\end{theorem}

\begin{proof}

\

\noindent\emph{The central choice and three binary variables.}

Assume for contradiction that such a solution $f$ exists. Consider the central profile
\begin{equation}
 X=((1,1,1,6),3),
 \qquad
 q(X)=\left(\frac13,\frac13,\frac13,2\right).
 \label{eq:X}
\end{equation}
Its quota allocations are $\alloc{1002}$, $\alloc{0102}$, and $\alloc{0012}$. We first analyze the branch
\begin{equation}
 f(X)=\alloc{1002}.
 \label{eq:branch}
\end{equation}
The other two choices will be handled by relabeling at the end.

Introduce the profiles in Table~\ref{tab:bits}. At each of $A,B,C$, quota initially permits three allocations. The excluded third option shown below the table is impossible relative to \eqref{eq:branch}: at $A$ and $B$ it transfers a seat from state 1 to state 2, and at $C$ it transfers a seat from state 1 to state 3; the two relevant populations remain equal to one. The lemma therefore leaves exactly two choices, which define bits $A,B,C\in\{0,1\}$.

\begin{table}[t]
\centering
\caption{Profiles that encode and copy the three Boolean variables. The two listed allocations define bit values 0 and 1.}
\label{tab:bits}
\scriptsize
\setlength{\tabcolsep}{2.6pt}
\renewcommand{\arraystretch}{1.08}
\begin{tabular}{@{}lccc@{}}
\toprule
Profile $(p;h)$ & exact quota pattern & bit 0 & bit 1\\
\midrule
$A=((1,1,1,9),8)$ & $(\frac23,\frac23,\frac23,6)$ & \alloc{1106} & \alloc{1016}\\
$B=((1,1,9,1),8)$ & $(\frac23,\frac23,6,\frac23)$ & \alloc{1160} & \alloc{1061}\\
$C=((1,9,1,1),8)$ & $(\frac23,6,\frac23,\frac23)$ & \alloc{1610} & \alloc{1601}\\
\midrule
$D_A=((2,1,1,6),5)$ & $(1,\frac12,\frac12,3)$ & \alloc{1103} & \alloc{1013}\\
$D_B=((2,1,6,1),5)$ & $(1,\frac12,3,\frac12)$ & \alloc{1130} & \alloc{1031}\\
$D_C=((2,6,1,1),5)$ & $(1,3,\frac12,\frac12)$ & \alloc{1310} & \alloc{1301}\\
\bottomrule
\end{tabular}

\vspace{2pt}
\parbox{0.96\columnwidth}{\scriptsize The excluded third options at $A,B,C$ are, respectively, \alloc{0116}, \alloc{0161}, and \alloc{0611}.}
\end{table}

The lower three profiles copy the bits to house size five. To see this for $A$ and $D_A$, compare their two listed options. The mismatch $A=0,D_A=1$ transfers a seat $2\to3$, whereas $A=1,D_A=0$ transfers a seat $3\to2$; populations of states 2 and 3 are one at both profiles. Both mismatches violate the lemma. The same argument applies to $(B,D_B)$ using states 2 and 4 and to $(C,D_C)$ using states 3 and 4. Consequently,
\begin{equation}
 D_A=A,\qquad D_B=B,\qquad D_C=C,
 \label{eq:copies}
\end{equation}
where a copy profile is assigned the option with the same bit index.

\noindent\emph{Six implication bridges.}

Each bridge profile in Table~\ref{tab:bridges} has total population 18 and house size 10. Its exact quota vector is a permutation of
\[
\left(\frac{10}{9},\frac{5}{9},\frac{10}{3},5\right),
\]
so it has exactly the three quota allocations displayed. The bold allocation is the target. For each bridge, one bit value makes the target impossible, while the negation of the desired conclusion, through the appropriate copy profile in \eqref{eq:copies}, makes both non-target allocations impossible. Hence the target would be simultaneously forbidden and forced.

\begin{table*}[t]
\centering
\caption{The six bridge gadgets. The notation $i\to j$ records a forbidden transfer between states whose populations are unchanged. In the ``non-target transfers'' column, the two transfers correspond in order to the two non-target quota allocations.}
\label{tab:bridges}
\scriptsize
\setlength{\tabcolsep}{3.6pt}
\renewcommand{\arraystretch}{1.13}
\begin{tabular}{@{}llllll@{}}
\toprule
Bridge $(p;h)$ & quota allocations & target forbidden by & target forced by & non-target transfers & implication\\
\midrule
$R_{AB}=((2,1,6,9),10)$ & \alloc{2035}, $\mathbf{1135}$, \alloc{1045} & $A=1:4\to2$ & $D_B=0$ & $2\to1,\ 2\to3$ & $A\Rightarrow B$\\
$R_{BA}=((2,1,9,6),10)$ & \alloc{2053}, $\mathbf{1153}$, \alloc{1054} & $B=1:3\to2$ & $D_A=0$ & $2\to1,\ 2\to4$ & $B\Rightarrow A$\\
$R_{AC}=((2,6,1,9),10)$ & \alloc{2305}, \alloc{1405}, $\mathbf{1315}$ & $A=0:4\to3$ & $D_C=0$ & $3\to1,\ 3\to2$ & $\neg A\Rightarrow C$\\
$R_{CA}=((2,9,1,6),10)$ & \alloc{2503}, $\mathbf{1513}$, \alloc{1504} & $C=1:2\to3$ & $D_A=1$ & $3\to1,\ 3\to4$ & $C\Rightarrow\neg A$\\
$R_{BC}=((2,6,9,1),10)$ & \alloc{2350}, \alloc{1450}, $\mathbf{1351}$ & $B=0:3\to4$ & $D_C=1$ & $4\to1,\ 4\to2$ & $\neg B\Rightarrow\neg C$\\
$R_{CB}=((2,9,6,1),10)$ & \alloc{2530}, \alloc{1540}, $\mathbf{1531}$ & $C=0:2\to4$ & $D_B=1$ & $4\to1,\ 4\to3$ & $\neg C\Rightarrow\neg B$\\
\bottomrule
\end{tabular}
\end{table*}

The bridge logic is therefore the six-clause system
\[
\begin{gathered}
A\Rightarrow B,\qquad B\Rightarrow A,\\
\neg A\Rightarrow C,\qquad C\Rightarrow\neg A,\\
\neg B\Rightarrow\neg C,\qquad \neg C\Rightarrow\neg B.
\end{gathered}
\]
The first pair gives $A=B$, the middle pair gives $A\ne C$, and the last pair gives $B=C$. For completeness, we verify every row directly.

\emph{The pair $A,B$.} At $R_{AB}$, if $A=1$, then comparing \alloc{1016} at $A$ with the target \alloc{1135} transfers a seat $4\to2$ while populations $p_4=9$ and $p_2=1$ are unchanged; the target is forbidden. If $B=0$, then $D_B=0$ gives \alloc{1130}. From $D_B$ to $R_{AB}$ only state 4's population changes. The non-target \alloc{2035} transfers $2\to1$, and \alloc{1045} transfers $2\to3$, so both are forbidden. Thus $B=0$ forces the target, proving $A=1\Rightarrow B=1$.

At $R_{BA}$, $B=1$ gives \alloc{1061} at $B$, and the target \alloc{1153} transfers $3\to2$ between unchanged populations, so it is forbidden. If $A=0$, then $D_A=0$ gives \alloc{1103}. Only state 3's population changes on moving to $R_{BA}$; the non-targets \alloc{2053} and \alloc{1054} transfer $2\to1$ and $2\to4$, respectively. Hence $A=0$ forces the target and $B=1\Rightarrow A=1$. Together,
\begin{equation}
 A=B.
 \label{eq:AB}
\end{equation}

\emph{The pair $A,C$.} At $R_{AC}$, $A=0$ gives \alloc{1106} at $A$; the target \alloc{1315} transfers $4\to3$ while populations $p_4=9$ and $p_3=1$ are unchanged, so the target is forbidden. If $C=0$, then $D_C=0$ gives \alloc{1310}. Only state 4's population changes on moving to $R_{AC}$; the non-targets \alloc{2305} and \alloc{1405} transfer $3\to1$ and $3\to2$. Thus $C=0$ forces the target and $A=0\Rightarrow C=1$.

At $R_{CA}$, $C=1$ gives \alloc{1601} at $C$; the target \alloc{1513} transfers $2\to3$ between unchanged populations, so it is forbidden. If $A=1$, then $D_A=1$ gives \alloc{1013}. Only state 2's population changes on moving to $R_{CA}$; the non-targets \alloc{2503} and \alloc{1504} transfer $3\to1$ and $3\to4$. Hence $A=1$ forces the target and $C=1\Rightarrow A=0$. The two implications are equivalent to
\begin{equation}
 C=1-A.
 \label{eq:AC}
\end{equation}

\emph{The pair $B,C$.} At $R_{BC}$, $B=0$ gives \alloc{1160} at $B$; the target \alloc{1351} transfers $3\to4$ between unchanged populations, so it is forbidden. If $C=1$, then $D_C=1$ gives \alloc{1301}. Only state 3's population changes on moving to $R_{BC}$; the non-targets \alloc{2350} and \alloc{1450} transfer $4\to1$ and $4\to2$. Thus $C=1$ forces the target and $B=0\Rightarrow C=0$.

At $R_{CB}$, $C=0$ gives \alloc{1610} at $C$; the target \alloc{1531} transfers $2\to4$ between unchanged populations, so it is forbidden. If $B=1$, then $D_B=1$ gives \alloc{1031}. Only state 2's population changes on moving to $R_{CB}$; the non-targets \alloc{2530} and \alloc{1540} transfer $4\to1$ and $4\to3$. Hence $B=1$ forces the target and $C=0\Rightarrow B=0$. Therefore
\begin{equation}
 B=C.
 \label{eq:BC}
\end{equation}

\noindent\emph{Contradiction and removal of the conditional branch.}

Equations \eqref{eq:AB}--\eqref{eq:BC} imply
\[
 A=B=C=1-A,
\]
which is impossible for Boolean variables. Therefore no solution satisfying both axioms can obey \eqref{eq:branch}.

It remains to avoid any hidden neutrality assumption. At $X$, quota gives the marginal seat to some $r\in\{1,2,3\}$. Choose a coordinate permutation $\pi$ that fixes state 4 and satisfies $\pi(r)=1$. For a vector $x\in\mathbb{R}^4$, define the action by
\[
 (\pi x)_i=x_{\pi^{-1}(i)}.
\]
Now define the conjugate solution
\begin{equation}
 f^{\pi}(p,h)=\pi\bigl(f(\pi^{-1}p,h)\bigr).
 \label{eq:conjugate}
\end{equation}
Quota is invariant under coordinate permutation. The inequalities and seat comparisons in \eqref{eq:population-paradox} are also merely relabeled, so $f^{\pi}$ is population monotone whenever $f$ is. The population vector of $X$ is invariant under permutations of its first three coordinates, whereas the chosen marginal winner is moved to state 1. Hence $f^{\pi}(X)=\alloc{1002}$, contradicting the conditional argument. This proves the theorem.
\end{proof}

\begin{corollary}
No deterministic four-state apportionment solution satisfies both quota and relative population monotonicity across house sizes.
\end{corollary}

\begin{proof}
Every forbidden transfer in the proof has $p_i'=p_i$ and $p_j'=p_j$, hence $p_i'/p_i=p_j'/p_j=1$. Each witness is therefore also a relative population paradox.
\end{proof}

\begin{corollary}
There is no randomized apportionment method whose realized allocation path is almost surely a globally defined four-state solution satisfying both quota and population monotonicity.
\end{corollary}

\begin{proof}
Such a method would be a probability distribution supported on deterministic solutions having both properties, but the theorem shows that the support set is empty.
\end{proof}

\section{Interpretation as constrained quantization}

At a fixed house size, apportionment has the geometry of quota-constrained integer rounding \cite{gray1998quantization}. The fractional quota vector lies on the simplex layer
\[
 \mathcal{S}_h=\{q\in\mathbb{R}_{\ge0}^4:\1^{\mathsf T}q=h\},
\]
and the reproduction alphabet is
\[
 \mathcal{A}_h=\{a\in\Znn^4:\1^{\mathsf T}a=h\}.
\]
Differences between reproduction points lie in the root lattice
\[
 A_3=\{z\in\mathbb{Z}^4:\1^{\mathsf T}z=0\}.
\]
Thus the ambient representation uses four coordinates, but both the simplex layer and the lattice have intrinsic dimension three.

A general solution $f(p,h)$ is most precisely a quota-constrained selector: absent homogeneity, it may use information in the integer population vector $p$ beyond the quota vector $q(p,h)$. If homogeneity is additionally imposed, the selector factors through $q$ and becomes an ordinary quantizer $Q_h:\mathcal{S}_h\to\mathcal{A}_h$.

Quota is a strong local error constraint: each coordinate is rounded to an adjacent integer, so $|a_i-q_i|<1$ whenever $q_i$ is nonintegral. More precisely, write $q=\ell+r$, with $\ell_i=\lfloor q_i\rfloor$ and $r_i\in[0,1)$. Since $\sum_i q_i=h$, the integer $k=\sum_i r_i$ is the number of coordinates that must be rounded upward. The quota allocations are exactly
\begin{equation}
 a=\ell+b,
 \qquad b\in\{0,1\}^4,
 \qquad \1^{\mathsf T}b=k,
 \qquad b_i=0\ \text{if }r_i=0.
 \label{eq:hypersimplex}
\end{equation}
When $k=2$ and all residues are positive, the six feasible binary vectors are the vertices of the hypersimplex $\Delta(4,2)$ and form the octahedral Johnson graph $J(4,2)$.

This perspective also recovers a classical method. Under squared Euclidean distortion,
\begin{equation}
 \|q-(\ell+b)\|_2^2
 =\|r\|_2^2+\sum_i b_i(1-2r_i).
 \label{eq:hamilton-objective}
\end{equation}
Subject to $\sum_i b_i=k$, a minimizer rounds upward the $k$ largest residues, which is Hamilton's method; ties can produce multiple minimizers and require a tie-breaking convention. Other apportionment methods arise from separable convex integer objectives and alternative error geometries \cite{gaffke2008vector,pukelsheim2017proportional}.

The theorem is therefore a global compatibility obstruction for deterministic, exact-sum quota selectors on four-coordinate inputs. It is not a conventional lower bound on average distortion. Without homogeneity, the result is not literally about a single map $q\mapsto a$; rather, no globally defined selector on population--house inputs can make all locally admissible quota choices compatible with the relational order constraint \eqref{eq:population-paradox} across population vectors and simplex layers. Under homogeneity, this becomes a no-go theorem for a family of quantizers $Q_h$. The gadget exposes the obstruction as a failure of global gluing: transporting local labels through the copy and bridge constraints creates a frustrated cycle that returns with the bit complemented.

\section{Conclusion and applications}

The result determines the first number of states at which incompatibility appears, under population monotonicity across possibly different house sizes, without anonymity, neutrality, order preservation, homogeneity, or other regularity assumptions: three states admit the Webster solution, whereas four states do not. The fifth state in the earlier construction of  G\"olz et al.\ \cite{golz2025lottery} is not intrinsically necessary. Its pigeonhole slack can be replaced by a finite logical network that synchronizes rounding decisions across profiles and house sizes. Because every witness uses equal population ratios, the same impossibility holds for relative population monotonicity.

For institutional design, the theorem sharpens a genuine trilemma. Already with four claimants, a deterministic system cannot simultaneously guarantee exact total allocation, lower-or-upper quota for every claimant, and the absence of the population paradox in Definition~1. A designer must instead prioritize monotonicity through a divisor-type method and tolerate some quota violations; preserve quota and accept that a paradox occurs somewhere; weaken the guarantee to an approximate or domain-restricted one; or use randomization with ex-ante or stochastic, rather than realization-by-realization, monotonicity \cite{grimmett2004stochastic,golz2025lottery,cembrano2025combinatorial,correa2026monotone}. Multiobjective formulations can make such tradeoffs explicit rather than treating a structurally unavoidable failure as a defect of implementation \cite{shechter2024congressional}.

The result also rules out randomized rules whose realized global allocation path is required to satisfy quota and population monotonicity almost surely; it does not rule out ex-ante, approximate, or high-probability relaxations.

The same logic applies beyond legislative seats whenever indivisible positions or resources are allocated approximately in proportion to changing claims while total capacity may also change: party-list seats, committee and cabinet positions, recurring service slots, or organizational resources. 

Three limitations identify natural next problems. First, the proof compares house sizes $3$, $5$, $8$, and $10$; as far as we know, the four-state problem without auxiliary regularity assumptions remains open when every monotonicity comparison is restricted to a common house size. Second, the profiles permit zero-seat allocations and do not directly settle systems imposing one seat per state. Third, an impossibility theorem identifies a worst-case obstruction but not its frequency or severity under empirical population distributions. Same-house gadgets, positive-minimum constructions, and distribution-aware rounding or quantization bounds would make the theory still more directly applicable.

\section*{AI usage disclosure}
ChatGPT 5.6-Sol was used in the process of formal proofs and writing.

\bibliographystyle{elsarticle-num}
\bibliography{references}

\end{document}